\documentclass[letterpaper, 10 pt, conference]{ieeeconf}  

\pdfoutput=1
\IEEEoverridecommandlockouts                              
\usepackage{xcolor}
\usepackage{amsmath}
\usepackage{amssymb}
 
\usepackage{amsthm}
\usepackage{amsfonts}
\usepackage{graphicx}      
\newtheorem{theorem}{Theorem}
\newtheorem{lemma}{Lemma}
\newtheorem{proposition}{Proposition}

\theoremstyle{definition}
\newtheorem{definition}{Definition}
\newtheorem{assumption}{Assumption}
\newtheorem{remark}{Remark}

\DeclareMathOperator*{\argmin}{arg\,min}

\newcommand{\real}{\mathbb{R}}
\newcommand{\realnonneg}{\mathbb{R}_{\ge 0}}

\newcommand{\until}[1]{[#1]}
\newcommand{\map}[3]{#1:#2 \rightarrow #3}
\newcommand{\setdef}[2]{\{#1 \; | \; #2\}}
\newcommand{\setdefb}[2]{\big\{#1 \; | \; #2\big\}}
\newcommand{\setdefB}[2]{\Big\{#1 \; \Big| \; #2\Big\}}

\renewcommand{\bf}{\mathbf{f}} 
\newcommand{\bg}{\mathbf{g}}

\newcommand{\bk}{\mathbf{k}}

\newcommand{\bp}{\mathbf{p}}

\newcommand{\bu}{\mathbf{u}}

\newcommand{\bx}{\mathbf{x}}

\newcommand{\bz}{\mathbf{z}}

\newcommand{\bphi}{\boldsymbol{\phi}}

\newcommand{\Cc}{\mathcal{C}}
\newcommand{\Dc}{\mathcal{D}}

\newcommand{\Jc}{\mathcal{J}}
\newcommand{\Kc}{\mathcal{K}}
\newcommand{\Oc}{\mathcal{O}}

\newcommand{\Uc}{\mathcal{U}}
\newcommand{\Rc}{\mathcal{R}}

\newcommand{\Sc}{\mathcal{S}}

\newcommand{\Xc}{\mathcal{X}}

\renewcommand{\baselinestretch}{.98}
\allowdisplaybreaks

\usepackage[compress]{cite}

\title{\LARGE 
\textbf{Steering with Contingencies: \\Combinatorial Stabilization and Reach-Avoid Filters}
}

\author{Yana Lishkova$^{1}$, Pio Ong$^{1}$, Sander Tonkens$^{2}$,
Sylvia Herbert$^{2}$, Aaron D.\ Ames$^{1}$%
\thanks{This research was supported by the Technology Innovation
Institute and the AFOSR Grant No.\ 113535-19668.}%
\thanks{$^{1}$California Institute of Technology, Pasadena, CA, USA.
{\tt\small \{ylishkova, pioong, ames\}@caltech.edu}}%
\thanks{$^{2}$University of California, San Diego, CA, USA.
{\tt\small \{stonkens, sherbert\}@ucsd.edu}}%
}

\begin{document}

\maketitle
\thispagestyle{empty}
\pagestyle{empty}


\begin{abstract}
In applications such as autonomous landing and navigation, it is often desirable to steer toward a target while retaining the ability to divert to at least $r$ (out of $p$) alternative sites if conditions change. In this work, we formalize this combinatorial contingency requirement and develop tractable control filters for enforcement. \emph{Combinatorial stabilization} requires asymptotic stability of a selected equilibrium while ensuring the trajectory remains within the safe region of attraction of at least $r$-out-of-$p$ candidates. To enforce this requirement, we use control Lyapunov functions (CLFs) to construct regions of attraction, which are combined combinatorially within an optimization-based filter. \emph{Combinatorial targeting} extends this framework to finite-horizon problems using Hamilton-Jacobi backward reach-avoid sets, accommodating shrinking reachable regions due to finite horizons or resource depletion. In both formulations, the resulting \emph{combinatorial stability filter} and \emph{combinatorial reach-avoid filter} require only $p\!+\!1$ constraints, preventing combinatorial blow-up and enabling safe real-time switching between targets. The framework is demonstrated on two examples where the filters ensure steering with contingency and enable safe diversion.
\end{abstract}

\section{Introduction}

Contingency planning is highly valuable in autonomy. A spacecraft may need to abort a landing and divert to an alternative site upon detecting a hazard. An autonomous vehicle may need to remain within reach of charging stations as its battery depletes. In such scenarios, the system must steer toward an active target while ensuring that at least $r$-out-of-$p$ candidate targets remain reachable at all times along its trajectory. Naive formulations of this requirement are computationally intractable, as they would need to account for every possible combination of $r$ among the $p$ candidates.

Control Lyapunov functions (CLFs) have long served as the primary tool for certifying and enforcing stability~\cite{Bacciotti05_LiapunovStability, Artstein83_CLF, Sontag89_CLF}, while control barrier functions (CBFs) provide a complementary framework for safety~\cite{Ames19_CBFs}. Together they guarantee both stability and safety through a single efficient optimization at each state with straightforward inclusion of additional safety CBF constraints~\cite{Ames19_CBFs}. Such formulations encode a boolean conjunction over all constraints, requiring every safety certificate to be satisfied simultaneously. Attempts to encode richer logical operations have brought issues such as non-smoothness, combinatorial blow-up, or restrictive compatibility requirements to the forefront of research (see~\cite{Glotfelter17_NonsmoothCBF, Tan22_CompatibilityCheckingCBFs, molnar2023composing}). The combinatorial CBF framework~\cite{Ong25_combinatorialCBF} addresses some of these issues by enabling $r$-out-of-$p$ safety certificate satisfaction with exactly $p$ smooth constraints, avoiding combinatorial blow-up while preserving the safe set. Subsequent work~\cite{Ong26_ComboBackupCBF} improves the compatibility of multiple safety constraints. However, it does not address the joint problem of stabilization and contingency planning.

\textcolor{black}{Hamilton-Jacobi reachability (HJR) analysis circumvents the compatibility issue of multiple target, obstacle, and control constraints by encoding them within a single reach-avoid value function~\cite{Margellos11_ReachAvoid, Fisac15_ReachAvoid}.} 
Variations of HJR can be used to construct functions that satisfy CLF/CBF-like conditions~\cite{Hirsch2025ViscosityCBF, Begzadic25_SafeResets_Learning_HJB}. Although a single reach-avoid value function can encode multiple obstacles and targets, its default formulation does not allow for steering toward a nominal target while preserving others as alternative feasible backups. Multi-target reach-avoid problems have been studied in temporal-logic and task-composition frameworks, including minimally intervening filter formulations (e.g.,~\cite{Chen18_STLReachability, Jiang24_TemporalLogicTrees_HJB, Chen25_MultiReachAvoid_HJB, sharpless2025dual}), however they focus on sequential, logic- or temporal-rule specified target visits rather than maintaining simultaneous feasibility of multiple alternatives. Reachability has also been used to construct filters that guarantee a single escape maneuver, a single reset region, or online alternative-target assessment (e.g.,~\cite{Leung20_SafetyAssurance_HJB,
Begzadic25_SafeResets_Learning_HJB, lishkova2024divert}). Mixed-integer programming (MIP) formulations offer a  general approach to encoding logical constraint combinations but scale poorly with the number of constraints, limiting real-time applicability~\cite{Bemporad99_MIP}. 

In this paper, we extend the combinatorial control certificate  framework of~\cite{Ong25_combinatorialCBF, Ong26_ComboBackupCBF} to both Lyapunov-based and Hamilton-Jacobi-based certificates, enabling the system to pursue a selected target while remaining within reach of multiple alternatives without combinatorial blow-up. We first focus on \emph{combinatorial stabilization}, where the objective is to asymptotically stabilize a selected equilibrium while guaranteeing the closed-loop trajectory remains within reach of at least $r$ alternative equilibria at all times. For this purpose we use sublevel sets of local control Lyapunov functions to define a \emph{combinatorial region of attraction}, i.e., the set of states from which at least $r$ of the $p$ equilibria can be reached. This enables the construction of a \emph{combinatorial stabilization filter} that simultaneously drives the system toward the selected equilibrium and prevents it from leaving this set, guaranteeing asymptotic stability of the selected equilibrium, forward invariance of the combinatorial region of attraction, and continuity of the resulting controller. 

We then introduce \emph{combinatorial targeting} for finite-horizon problems where the set of states from which a target can be safely reached shrinks over time due to resource depletion. The goal is to reach a selected target within a given horizon while remaining within reach of at least $r$ candidate targets at all times. For this purpose, we construct a \emph{combinatorial backward reach-avoid set} from Hamilton-Jacobi reach-avoid value functions associated with each candidate target. Formulating an optimization-based \emph{combinatorial reach-avoid filter}, we prove that the selected target is reached while remaining in the combinatorial set throughout. This formulation is suited to problems where a CLF is difficult to obtain, though  grid-based HJR methods inherit the curse of dimensionality. In both formulations, the resulting controllers can be implemented as a quadratic program (QP) with $p\!+\!1$ constraints at each time step, admitting real-time online implementation and guaranteeing safe switching in the event of a target change. \textcolor{black}{We demonstrate the proposed framework on two examples, showing that the filters steer toward the selected target, while remaining in reach of $r$-out-of-$p$ others and enable safe target switching in scenarios where the nominal controller alone violates safety.}

\section{Background}
\noindent \textit{Notation:}
 Let $[p]:\!=\!\{1,\dots,p\}$ and $C^1$ denote the set of continuously differentiable functions. A continuous function $\alpha:[0,a)\!\to\![0,\infty)$ is class-$\mathcal K$ if it is strictly increasing and $\alpha(0)\!=\!0$. A continuous function $\gamma:(-b,a)\!\to\!\mathbb{R}$, $a,b\!>\!0$, is extended class-$\mathcal K$ ($\gamma\!\in\!\mathcal K_e$) if it is strictly increasing and $\gamma(0)\!=\!0$. For a $C^1$ function $h:\mathbb{R}^n \!\to\! \mathbb{R}$ and a vector field $f:\mathbb{R}^n \!\to\! \mathbb{R}^n$, the Lie derivative of $h$ along $f$ is defined as $L_{\bf} h(\bx)\! :=\! \frac{\partial h}{\partial x}(\bx)\, \mathbf{f}(\bx)
\!=\! \nabla h(\bx)^\top \bf(\bx)$. For $r>0$, $B_r(\bx^\star) \!:=\! \{\bx\in\mathbb{R}^n \mid \|\bx\!-\!\bx^\star\|<r\}$ denotes the open ball centered at $\bx^\star$. Additionally, $\mathrm{ReLU}(s) \!:=\! \max\{0,\, s\}$. 

\subsection{Control Lyapunov functions}
Consider the nonlinear control-affine system
\begin{equation}
    \dot{\bx} = \bf(\bx) + \bg(\bx)\bu,
    \label{sys:ctrl_affine}
\end{equation}
where $\bx\!\in\! \mathbb{R}^n$ is the state, $\bu \!\in\! \mathbb{R}^m$ the control input, and the maps $\map{\bf\!}{\!\!\real^n\!}{\!\real^n}$, $\map{\bg\!}{\!\!\real^n\!}{\!\real^{n\times m}}$ are assumed to be continuous. Let $\bx^\star$ be a (controlled) equilibrium of system~\eqref{sys:ctrl_affine}, i.e., there exists an $\bu^\star\!\in\!\mathbb{R}^m$ such that $\bf(\bx^\star)\!+\!\bg(\bx^\star)\bu^\star\! =\!0$. To design feedback controllers that stabilize the equilibrium $\bx^\star$, we employ \emph{control Lyapunov functions}~\cite{Sontag83_CLF, Bacciotti05_LiapunovStability}.

\begin{definition}[Local control Lyapunov function]
\label{def:CLF}
A $C^1$ function $V:\real^n\to\mathbb{R}_{\geq 0}$, is called a \emph{local control Lyapunov function (CLF)} for~\eqref{sys:ctrl_affine} about $\bx^\star$ if $V(\bx^\star)=0$ and there exist class-$\mathcal{K}$ functions $\alpha_1$, $\alpha_2$, and $\alpha$ such that
\begin{subequations}\label{eq:clf_def}
\begin{equation}
\alpha_1(\|\bx\!-\!\bx^\star\|)  \le V(\bx) \le \alpha_2(\|\bx\!-\!\bx^\star\|), \;\;\forall \bx\!\in \!B_r(\bx^\star), 
\end{equation}
\begin{equation}
\inf_{\bu\in\mathbb{R}^m}\dot V(\bx,\bu) \le -\alpha(V(\bx)), \quad\;\;\forall \bx\!\in\! B_r(\bx^\star)\!\setminus\!\{\bx^\star\}, 
\end{equation}
\end{subequations}
for some $r>0$, where $\dot V(\bx,\bu)\triangleq L_{\bf}V(\bx)+L_{\bg}V(\bx)\bu.$~\hfill$\diamond$
\end{definition}

CLFs extend the classical notion of Lyapunov functions to control settings. In particular, conditions~\eqref{eq:clf_def} ensure that, $\forall\bx\!\in\! B_r(\bx^\star)\!\setminus\!\{\bx^\star\}$, there exists a $\bu\!\in\!\mathbb R^m$ such that $\dot V(\bx,\bu) \!<\! 0$, as required in standard Lyapunov stability theory \cite{Bacciotti05_LiapunovStability}.
\begin{lemma}\label{lem:stab}
    Let $V$ be a local CLF for system~\eqref{sys:ctrl_affine}, and let $\map{\bk}{\real^n}{\real^m}$ be  a continuous feedback controller such that
    \begin{equation}\label{eq:lyap_cond}
    \dot V(\bx,\bk(\bx)) \leq -\alpha(V(\bx)), 
    \end{equation}
    and  $\bk(\bx^\star)\!=\!\bu^\star$ for all $\bx\!\in\! B_r(\bx^\star)\!\setminus\!\{\bx^\star\}$. Then the controller 
    $\bk$ asymptotically stabilizes the equilibrium $\bx^\star$.~\hfill$\blacksquare$
\end{lemma}

Lyapunov methods do not only certify stability, but also provide a characterization of the region of attraction. 
\begin{lemma}
\label{lem:CLFinvar}
Consider the same setting as in Lemma~\ref{lem:stab}. For $c>0$, define the sublevel set $\Omega_c:=\left\{\bx\in\real^n \mid V(\bx)\leq c\right\}$. Then for any $c >0$ such that $\Omega_c\subset B_r(\bx^\star)$, the state feedback $\bu=\bk(\bx)$ renders the set $\Omega_c$ forward invariant under the closed-loop dynamics. Moreover, $\Omega_c$ is a certified inner approximation of the region of attraction of $\bx^\star$.~\hfill$\blacksquare$
\end{lemma}

CLFs facilitate the design of a continuous feedback controller. In particular,
optimization-based approaches have been widely adopted for the controller synthesis, formulating the controller as a quadratic program (QP):
\begin{align}
    \label{eq:CLF_QP}
        \bk(\bx) =\argmin_{\bu \in \mathbb{R}^m} \quad
            & \|\bu - \bk_{\mathrm{nom}}(\bx)\|^2 \\
        \mathrm{s.t.} \quad
            & \dot{V}(\bx,\bu) \;\leq\; -\alpha(V(\bx)),\notag 
\end{align}
where $\bk_{\mathrm{nom}}$ is a nominal controller \cite{Ames19_CBFs}. The main benefit of this formulation is that it naturally accommodates additional constraints such as those encoding safety requirements.

\subsection{Control barrier functions} \label{section:CBFsTheory} 
While CLFs guarantee convergence to an equilibrium, control barrier functions guarantee the state remains within a safe set $\Sc\subseteq\real^n$. Since rendering all of $\Sc$ forward invariant is generally impossible, one can instead identify a subset $\Cc\subseteq \Sc$ in which the system can be initialized and operated safely, i.e., trajectories starting in $\Cc$ remain in $\Cc$ for all time.

To this end, we construct a candidate safe set as the $0$-superlevel set of a $C^1$ function $h:\mathbb{R}^n\to\mathbb{R}$, namely:
\begin{align}
     \mathcal{C} &= \{\bx\in\mathbb{R}^n \mid h(\bx)\ge 0\}.  \label{eqn:C}
\end{align}
To ensure safety, we choose a set $\Cc\subseteq\Sc$ and require it to be forward invariant under the closed-loop dynamics. This can be achieved by imposing an affine constraint on the evolution of $h$, motivating the notion of a control barrier function.
\begin{definition}[Control Barrier Function~{\cite{Ames19_CBFs,Ong25_combinatorialCBF}}]
\label{def:CBF}
A continuously differentiable function $h:\mathbb{R}^n\to\mathbb{R}$ is a \emph{control barrier function (CBF)} for~\eqref{sys:ctrl_affine} on $\mathcal{C}$ if there exists $\alpha\in\mathcal{K}_e$ such that:
\begin{equation}
\sup_{\bu\in\mathbb{R}^m}(\dot h(\bx,\bu))
> -\alpha(h(\bx))
\label{eq:cbf_def}
\end{equation}
for all $\bx\in\mathcal{C}$, where $\dot h(\bx,\bu)\triangleq L_\bf h(\bx)+ L_\bg h(\bx)\bu$.~\hfill$\diamond$
\end{definition}
This definition parallels that of a CLF, with the condition on $\dot h$ playing a role on ensuring the existence of a safeguarding controller, as formalized in the following lemma. 
\begin{lemma}\label{lem:CBFinvar}
Let $h:\mathbb{R}^n\to\mathbb{R}$ be a continuously differentiable function defining the $0$-superlevel set of a set $\mathcal{C}$ as in \eqref{eqn:C}, and let $\bk:\mathbb{R}^n\to\mathbb{R}^m$ be a continuous controller satisfying:
\begin{equation}\label{eq:cbf_cond}
    \dot h(\bx, \bk(\bx)) \geq -\alpha(h(\bx)),
\end{equation}
for all $\bx$ in a neighborhood $\Dc$ of $\Cc$. Then the state feedback $\bu=\bk(\bx)$ renders the set $\mathcal{C}$ forward invariant. \hfill$\blacksquare$
\end{lemma}

In practice, such a condition is often enforced via a QP-based safety filter of the form:
\begin{align}
    \label{eq:cbf_qp}
        \bk(\bx) = \argmin_{\bu \in \mathbb{R}^m} \quad
            & \tfrac{1}{2}\|\bu-\bk_{\mathrm{nom}}(\bx)\|^2 \\
        \text{s.t.}\quad
            & \dot{h}(\bx,\bu) \ge -\alpha(h(\bx)). \notag
\end{align}
Under the use of a strict inequality in Definition~\ref{def:CBF}, the CBF-QP is continuous and is well-defined on a neighborhood of $\Cc$ as required. This formulation is particularly useful when multiple safety requirements, and potentially a CLF condition, must be enforced simultaneously, leading to multiple affine constraints in the optimization.

\subsection{Combinatorial control barrier functions}

Given $p$ CBFs~$\{h_j\}_{j\in\until{p}}$, appending the associated CBF constraints to the CBF-QP enforces their simultaneous satisfaction, corresponding to an AND composition of safety requirements. Defining the corresponding $0$-superlevel sets as
$\mathcal{C}_j=\{\bx\in\mathbb{R}^n\mid h_j(\bx)\ge 0\}$ for all $j\in\until p$, the CBF-QP enforces forward invariance of their intersection:
\begin{equation}
     \textstyle\bigcap_{j=1}^p \mathcal{C}_j = \setdefb{\bx\in\real^n}{\min_{j\in[p]} h_j(\bx)\ge 0}
\end{equation}

More generally, one may wish to only require that at least $r$ of the $p$ functions $\{h_j(\bx)\}_{j\in\until p}$ are nonnegative at any given time. The corresponding combinatorial safe set is
\begin{equation}
    \tilde{\mathcal{C}}=\{\bx\in\mathbb{R}^n\mid \tilde h(\bx)\ge 0\}.
    \label{eqn:CkOne}
\end{equation}
where $\map{\tilde h}{\real^n}{\real}$ is the pivot function, defined as:
\begin{equation}
\tilde h(\bx) := \textstyle\max^{(r)} \{h_j(\bx)\}_{j\in\until p}
\label{eqn:sorting}
\end{equation}
and $\max^{(r)}_{j\in[p]}$ returns the $r$-th largest value among the collection $\{h_j(\bx)\}_{j\in\until p}$. 
The AND and OR cases can be recovered with $r\!=\!p$ and $r\!=\!1$, respectively. However, the resulting function $\tilde h$ is generally nonsmooth, which may lead to discontinuities in the corresponding control laws \cite{Ong25_combinatorialCBF}.

To handle such logical compositions, we rely on the combinatorial CBF framework by~\cite{Ong25_combinatorialCBF,Ong26_ComboBackupCBF}. In particular, $\tilde \Cc$ can be rendered forward invariant using the following result.
\begin{lemma}\label{lem:combo_CBF}
    Consider the set $\tilde \Cc$ constructed from multiple $C^1$ functions $\{h_j\}_{j\in\until p}$ as in~\eqref{eqn:CkOne}.
    Suppose there exist a continuous controller $\bk:\mathbb{R}^n\!\to\!\mathbb{R}^m$  and a nonnegative auxiliary function $\map{\theta}{\real^n}{\realnonneg}$ such that:
\begin{equation}
\dot{h}_j(\bx, \bk(\bx)) \geq -\alpha\bigl(h_j(\bx)\bigr)
- \theta(\bx)\rho\bigl(h_j(\bx) - \tilde{h}(\bx)\bigr),
\label{eq:combo_cbf_condition}
\end{equation}
for all $j\!\in\!\until{p}$ and all $\bx$ in a neighborhood $\Dc$ of $\tilde\Cc$, with  a positive definite function $\map{\rho}{\!\real\!}{\!\realnonneg\!}$ and $\alpha\in\Kc_e$. Then the controller $\bk$ renders set $\tilde \Cc$ forward invariant. \hfill$\blacksquare$
\end{lemma}
The controller $\bk$ and the auxiliary function $\theta$ can be simultaneously obtained by solving the following QP:
\begin{align}
    \label{eq:combocbf_qp}
        &{\small\begin{bmatrix}
            \bk(\bx)\\\theta(\bx)
        \end{bmatrix}} = \argmin_{\bu \in \mathbb{R}^m,~\omega\in\realnonneg} \;
             \tfrac{1}{2}\|\bu-\bk_{\mathrm{nom}}(\bx)\|^2 +c_\omega\omega^2 \\
        &\text{s.t.}\,\,
            \dot h_j(\bx,\bu) \ge \!-\alpha(h_j(\bx)) -\omega\rho\bigl(h_j(\bx) - \tilde{h}(\bx)\bigr),~\forall j\!\in\!\until p.\nonumber
\end{align}
The additional term involving $\rho$ introduces flexibility in the barrier condition, allowing some of the functions $h_j$ to become negative while still ensuring that at least $r$ of them remain nonnegative. The auxiliary variable $\omega$ scales this relaxation, enabling the controller to disregard noncritical constraints when necessary to maintain feasibility.

\section{Combinatorial Stabilization} \label{section:proposed}

Motivated by autonomy applications in which fault tolerance and contingency planning are essential, we introduce a control objective that extends classical notions of stabilization and backward reachability. In such settings, it is not sufficient to ensure that the system can be steered to a single target. Instead, the system must also remain controllable to multiple candidate targets throughout its trajectory, ensuring that the loss of any one option
does not compromise mission success and providing the system with contingency options for unexpected condition changes. For example, in spacecraft or drone landing, the system must retain the ability to divert to alternative landing sites in response to newly detected hazards. Similarly, in autonomous exploration, the system must remain controllable to charging stations or safe zones. 

In this section, we focus on stabilization and assume that a collection of stabilizing controllers $\{\bk_j\}_{j\in[p]}$ is given, each associated with a distinct equilibrium $\{\bx_j^\star\}_{j\in[p]}$. For each $j \in [p]$, let $\Rc_j$ be the certified subset of region of attraction under $\bk_j$. Our objective is to ensure \emph{combinatorial stabilization}, i.e., stabilization while remaining within a \emph{combinatorial region of attraction}. We formalize these notions as follows.
\begin{definition}
    Given a set of equilibria $\{\bx^\star_j\}_{j\in\until p}$ under controllers $\{\bk_j\}_{j\in[p]}$, let $\{\Rc_j\}_{j\in[p]}$ be subsets of their corresponding regions of attraction. With $r\in\until p$, the set:
\begin{equation}\label{eq:combo_RoA}
\tilde{\Rc}:=\setdefB{\bx \in \real^n}{\big|\setdef{j\in\until p}{\bx\in \Rc_j}\big|\geq r}
\end{equation} 
is called a \emph{($r$-out-of-$p$) combinatorial region of attraction}.\hfill$\diamond$
\end{definition}

\begin{definition}
    A controller $\map{\!\bk\!}{\!\real^n\!}{\!\real^m\!}$ \emph{($r$-out-of-$p$) combinatorially stabilizes} an equilibrium $\bx^\star_{j^\dagger}\!\in\!\{\bx_j^\star\}_{j\in\until p}$ if, under the feedback $\bu\!=\!\bk(\bx)$, the set $\tilde\Rc$ is forward invariant and $\bx^\star_{j^\dagger}$ is asymptotically stable with $\tilde\Rc\! \cap\!\Rc_{j^\dagger}$ contained in its region of attraction.\hfill$\diamond$
\end{definition}

\subsection{Combinatorial Stabilization via CLFs}\label{section:CLF}
We now show how combinatorial stabilization above can be achieved using control Lyapunov functions. Consider the control affine system in~\eqref{sys:ctrl_affine} with multiple controlled equilibria $\{\bx^\star_j\}_{j\in\until p}$ of interest. We assume that for each equilibrium $\bx_j^\star$, there exists a local control Lyapunov function $\map{V_j}{\real^n}{\realnonneg}$ for~\eqref{sys:ctrl_affine} about $\bx^\star_j$. As discussed in Lemma~\ref{lem:CLFinvar}, for each equilibrium $\bx^\star_j$, there exists $c_j>0$ such that the sublevel set 
$$
\Rc_j=\setdefb{\bx\in\real^n}{V_j(\bx)\leq c_j}
$$
is a certified inner approximation of the region of attraction of $\bx_j^\star$ (under any controller $\bk_j$ satisfying~\eqref{eq:lyap_cond}).  Within this setup, our objective reduces to ensuring forward invariance  of the combinatorial region of attraction $\tilde{\Rc}$ constructed from $\{\Rc_j\}_{j\in\until p}$ in~\eqref{eq:combo_RoA} while stabilizing to a chosen equilibrium. 

To achieve this, we construct CBFs from the given CLFs and combine them using the combinatorial CBF framework. By defining the CBFs:
\begin{equation}
h_j(\bx) := c_j - V_j(\bx), \qquad j \in [p],
\label{eq:hk_def}
\end{equation}
each $\Rc_j$ is the $0$-superlevel set of $h_j$. The combinatorial region of attraction can then be characterized as:
\begin{equation}\label{eq:combo_RoA_from_CBFs}
   \tilde \Rc = \setdefb{\bx\in\real^n}{\tilde h(\bx)\geq 0}. 
\end{equation}
where $\tilde h(\bx)=\max^{(r)} \{h_j(\bx)\}_{j\in\until p}$. With this construction, we combine combinatorial CBF and CLF constructions to establish combinatorial stabilization.

A natural approach is to impose both the Lyapunov decrease condition~\eqref{eq:lyap_cond} for a designated equilibrium $\bx^\star_{j^\dagger}$ and the combinatorial barrier condition~\eqref{eq:combo_cbf_condition} simultaneously. The issue with such an approach is the lack of feasibility. In particular, recall that the local control Lyapunov function~\eqref{eq:clf_def} only guarantees admissible control inputs in a small neighborhood~$B_r(\bx^\star_{j^\dagger})$ of the equilibrium. As such, we propose an adjustment to the Lyapunov condition, to ensure its feasibility within all of $\tilde \Rc$, in the following result.

\begin{theorem}\label{thm:combo_CLF_CBF}
    Consider the control-affine system~\eqref{sys:ctrl_affine} with a collection of local CLFs $\{V_j\}_{j\in\until p}$ about corresponding equilibria $\{\bx_j^\star\}_{j\in\until p}$. Let the functions
$\{h_j\}_{j\in\until p}$ be defined as in~\eqref{eq:hk_def}, and $\tilde \Rc$ as in~\eqref{eq:combo_RoA}. Fix an index $j^\dagger\in\until p$. Suppose the equilibrium $\bx^\star_{j^\dagger}$ is in the interior of $\tilde \Rc \cap \Rc_{j^\dagger}$, and there exists a continuous controller $\map{\bk}{\real^n}{\real^m}$ and an auxiliary function $\map{\theta}{\real^n}{\realnonneg}$ such that:
\begin{subequations} \label{eq:combo_CLF_CBF}
    \begin{align} 
\!\!\!\dot V_{j^\dagger}(\bx,\bk(\bx)) &\!\leq\! -\alpha_{j^\dagger}(V_{j^\dagger}(\bx))\!+\! \textcolor{black}{\theta(\bx)\mathrm{ReLU}\bigl(\!-h_{j^\dagger}(\bx)\bigr)} \label{eq:combo_CLF_CBF-CLF}\\  \!\!\!
        \dot h_j(\bx,\bk(\bx)) &\!\geq\! -\alpha_{\tilde h}\bigl(h_j(\bx)\bigr)
  \!-\!\theta(\bx)\rho\bigl(h_j(\bx)\!-\!\tilde h(\bx)\bigr), \label{eq:combo_CLF_CBF-CBF}
    \end{align}
\end{subequations}
for all $j\!\in\!\until p$ and all $ \bx$ in a neighborhood $\Dc$ of $\tilde \Rc$, where 
$\alpha_{\tilde h}\!\in\!\Kc_e$ and $\map{\rho\!}{\!\real\!}{\!\realnonneg}$ is positive definite. Then, the controller $\bk$ combinatorially stabilizes the equilibrium $\bx^\star_{j^\dagger}$.
\end{theorem}
\begin{proof}
    Forward invariance of $\tilde \Rc$ follows directly from Lemma~\ref{lem:combo_CBF}. We then seek to show that $\bx^\star_{j^\dagger}$ is asymptotically stable with $\tilde \Rc\cap \Rc_{j^\dagger}$ contained in its region of attraction.

    If $\bx\!\in\!\tilde \Rc\cap \Rc_{j^\dagger}\subseteq\Rc_{j^\dagger}$, then $V_{j^\dagger}(\bx)\!\leq\! c_{j^\dagger}$ and the $\mathrm{ReLU}$ function in~\eqref{eq:combo_CLF_CBF-CLF} evaluates to zero. Thus, the Lyapunov decrease condition~\eqref{eq:lyap_cond} holds $\forall \bx \in \tilde \Rc\cap \Rc_{j^\dagger}$ making the set forward invariant. In particular, this holds on a neighborhood $\mathcal{N}\!\subseteq\! \tilde \Rc \cap \Rc_{j^\dagger}$ of the $\bx^\star_{j^\dagger}$, which exists by assumption. Asymptotic stability of the equilibrium then follows from~Lemma~\ref{lem:stab}. Next, we prove convergence to the equilibrium from any initial condition $\bx_0\!\in \!\tilde \Rc\!\cap\! \Rc_{j^\dagger}$. Let $t\!\mapsto \!\bx(t)$ be any solution to the closed-loop system, it satisfies: 
    $$
    \dot V_{j^\dagger}(\bx(t),\bk(\bx(t))) \leq -\alpha_{j^\dagger}(V_{j^\dagger}(\bx(t))). 
    $$
    From the comparison lemma, $V_{j^\dagger}(\bx(t))\!\rightarrow \!0$, implying that $\bx(t)\!\rightarrow\! \bx^*_{j^\dagger}$. Hence, $\tilde \Rc\cap \Rc_{j^\dagger}$ is contained in the region of attraction of $\bx^\star_{j^\dagger}$ under $\bk$.
\end{proof}

Theorem~\ref{thm:combo_CLF_CBF} provides conditions to achieve combinatorial stabilization by combining CLF and combinatorial CBF constraints. The relaxation term in~\eqref{eq:combo_CLF_CBF-CLF} activates only outside the certified region of attraction~$\Rc_{j^\dagger}$, where an admissible input satisfying the Lyapunov condition may not exist. Consequently, nominal CLF-based stabilization is preserved for trajectories initialized in $\tilde\Rc \cap \Rc_{j^\dagger}$. Moreover, this relaxation ensures that the controller remains well-defined over the entire combinatorial region $\tilde \Rc$, even when the selected equilibrium is not locally stabilizable. This lays the foundation for future robustness analyses against disturbances, model errors, or incorrect target selections.

\subsection{Combinatorially stabilizing controllers}
The constraint structure in~\eqref{eq:combo_CLF_CBF}  naturally leads to an optimization-based controller synthesis. In particular, we construct a QP-based \textbf{combinatorial stabilization filter}: 
\begin{align}
        &\!\!\!\!\!
        {\small\begin{bmatrix}
            \bk(\bx,\tau_1,\tau_2) \\
            \theta(\bx,\tau_1,\tau_2) 
        \end{bmatrix}}\!= \argmin_{\substack{\bu \in \mathbb{R}^m,\,\omega\in\realnonneg}}  \!\tfrac{1}{2}\|\bu\!-\!\bk_{\mathrm{nom}}(\bx)\|^2 \!+\!c_\omega \boldsymbol{\omega} ^2 \label{eq:combo_clf_cbf_qp}\\ 
        &
       \text{s.t.}\, \dot V_{j^\dagger}(\bx,\bu) \leq\! -\alpha_{j^\dagger}(V_{j^\dagger}(\bx))\!+\!\omega\textcolor{black}{\mathrm{ReLU}\bigl(-h_{j^\dagger}(\bx)\bigr)}\nonumber\\
            &\dot h_j(\bx,\bu) \ge\! -\alpha_{\tilde h}(h_j(\bx))\! -\omega\rho\bigl(h_j(\bx) - \tilde{h}(\bx)\bigr),~\forall j\in\until p,\nonumber
\end{align}
which satisfies the required constraints in~\eqref{eq:combo_CLF_CBF} by design. Nevertheless, continuity of the controller is desirable in general for well-posedness of the closed-loop system, e.g., to avoid chattering. This property is ensured for the QP~\eqref{eq:combo_clf_cbf_qp} under standard regularity conditions, as stated next.

\begin{proposition}\label{prop:slater}
    Consider the same setting as in Theorem~\ref{thm:combo_CLF_CBF}. Define the index set of critical CBFs as:
    $$
    \Jc(\bx) = \setdefb{j\in\until p}{h_j(\bx)=\tilde h(\bx)}.
    $$
    Suppose the following hold:
    \begin{enumerate}
        \item For each $\bx\in\tilde\Rc\cap\Rc_{j^\dagger}$, there exists $\bu\in\real^m$ such that:
        \begin{align*}
        \dot V_{j^\dagger}(\bx,\bu) &< -\alpha_{j^\dagger}(V_{j^\dagger}(\bx)) \nonumber\\
            \dot h_j(\bx,\bu) &> -\alpha_{\tilde h}(h_j(\bx)),~\forall j\in\Jc(\bx);
    \end{align*}
    \item For each $\bx\in \tilde\Rc\setminus\Rc_{j^\dagger}$, there exists $\bu\in\real^m$ such that:
        \begin{align*}
            \dot h_j(\bx,\bu) &> -\alpha_{\tilde h}(h_j(\bx)),~\forall j\in\Jc(\bx).
    \end{align*}
    \end{enumerate}
    Then the combinatorially stabilizing QP~controller $\bk$ defined in~\eqref{eq:combo_clf_cbf_qp} is continuous on a neighborhood $\Dc$ of $\tilde\Rc$. Consequently, the controller $\bk$ combinatorially stabilizes $\bx^\star_{j^\dagger}$. 
\end{proposition}
\begin{proof}
    We aim to show that the QP~\eqref{eq:combo_clf_cbf_qp} satisfies a pointwise Slater's condition on a neighborhood $\Dc$ of $\tilde\Rc$. In particular, for each $\bx \in \Dc$, there exists $(\bu,\omega)\in\real^m\times \realnonneg$ such that all constraints are satisfied strictly. 

    Fix $\bx\in \tilde \Rc$. By assumption, there exists a control $\bu$ such that the constraints with $\rho(\cdot)= 0$ or $\mathrm{ReLU}(\cdot)= 0$ are met strictly. For the remaining constraints, for which $\rho(\cdot)>0$ or $\mathrm{ReLU}(\cdot)>0$, a sufficiently large $\omega$ can be selected to render these finitely many constraints strictly satisfied. Thus, there exists $(\bu,\omega)$ that strictly satisfies all constraints at $\bx$, establishing the pointwise Slater's condition on $\tilde \Rc$. In addition, as all functions involved in the constraints are continuous, strict feasibility persists on a neighborhood $\Dc$ of $\tilde\Rc$. Continuity of the QP controller~\eqref{eq:combo_clf_cbf_qp} then follows from satisfying Slater's condition at each $\bx\in\Dc$~\cite{PM-AA-JC:25}. 
\end{proof}
Proposition~\ref{prop:slater} characterizes the pointwise Slater's condition that provides the standard sufficient condition ensuring the continuity of the QP controller. In particular, due to the presence of the auxiliary variable $\omega$, strict feasibility of the constraints reduces to ensuring strict feasibility only for the critical constraints indexed by $\Jc(\bx)$.
\begin{remark}
    The conditions in Proposition~\ref{prop:slater} can be further simplified under the choice of $\alpha_{\tilde h}$. In particular, since the set $\tilde\Rc$ is compact (as it is defined through sublevel sets of Lyapunov functions), there exists a sufficiently large function $\alpha_{\tilde h}$ such that the strict feasibility condition need only be verified on the boundary of the set~$\tilde \Rc$.~\hfill$\bullet$
\end{remark}

\subsection{Discussion on results}
The combinatorial stabilization framework provides robustness at the planning level by ensuring that the closed-loop state remains within the region of attraction of at least $r$ equilibria  at all times. A key consequence of this property is the ability to switch between equilibria in real time. At each state $\bx\in\tilde\Rc$, one may choose to stabilize to any of the $r$ (or more) equilibria by appropriately adjusting $j^\dagger$. This enables contingency planning and reactive decision-making, allowing the system to adapt to changing conditions without compromising stability guarantees. Switching between targets may introduce discontinuities in the control input due to changes in the nominal objective. However, for any fixed choice of $j^\dagger$, the resulting controller is continuous. When switching occurs at discrete times, the resulting piecewise continuous control signal remains well-posed (assuming there are not infinite number of switches in finite time), with each switching instant treated as a new initial condition.

Furthermore, the optimization-based formulation readily accommodates additional safety constraints through CBFs, highlighting the flexibility of the proposed framework for combinatorial stabilization. While the integration of CBFs is not the primary focus of this work, it is worth noting that, in practice, feasibility is often maintained by introducing slack variables that relax the Lyapunov condition. Care must be taken, however, as such relaxations can compromise (combinatorial) stability guarantees. In this work, safety constraints are instead handled through the Hamilton-Jacobi reach-avoid framework, which incorporates both safety and reachability (or stability) within a single value function.

\section{Combinatorial Backward Reachability}
\label{section:HJR}

The CLF-based construction of Section~\ref{section:CLF} requires a local CLF for each candidate target, which may be difficult to obtain for general nonlinear systems. Moreover, known constructive methods for obtaining such functions typically produce regions of attraction that are infinite-horizon in nature,
and thus do not capture settings where reachable regions shrink over time due to fuel depletion, battery discharge, or an approaching mission deadline.
HJR analysis provides a systematic alternative by computing viscosity CLFs/CBFs~\cite{Hirsch2025ViscosityCBF}. More generally, HJR constructs reach-avoid value functions directly while naturally accommodating finite time horizons, multiple obstacles, and control constraints~\cite{Margellos11_ReachAvoid, Fisac15_ReachAvoid}. This section develops an HJR-based analogue to the combinatorial stabilization framework by replacing the CLF-induced regions of attraction $\Rc_j$ with reach-avoid sets derived from HJ value functions.

HJR analysis aims to characterize backward reach-avoid (BRA) sets, i.e., sets of initial conditions from which there exists a measurable control signal that steers the system to a target set within a given time horizon while avoiding unsafe states. 
We first recall the standard definition.

\begin{definition}[Backward Reach-Avoid Sets {\cite{mitchell2005time, Margellos11_ReachAvoid}}]
\label{def:BRAS}
Consider~\eqref{sys:ctrl_affine} with a target set $\mathcal{X}^\star_j\subset \real^n$, an obstacle set $\mathcal{O}\subset\real^n$, and the control constraint set $\Uc\subseteq\real^m$. Given a negative time horizon $\tau<0$, a measurable control signal $\map{\bu}{[\tau,0]}{\Uc}$, and an initial condition $\bx\in\real^n$, let $\map{\bx_\bu(\cdot;\bx)}{[\tau,0]}{\real^n}$ denote the corresponding state trajectory.  The time-varying \emph{backward reach-avoid (BRA) set} is defined as:
\begin{align*}
    \mathrm{BRA}(\mathcal{X}^\star_j,& \mathcal{O}, \tau) := \bigl\{
    \bx \in \mathbb{R}^n \;\big|\;
    \exists\, \bu(\cdot) \in \mathbb{U}~\text{and} ~s \in [\tau, 0]\\
    \text{s.t.}~&\bx_{\bu}(s; \bx) \in \mathcal{X}^\star_j \;\text{and}\;
    \bx_{\bu}(\sigma; \bx) \notin \mathcal{O},\;
    \forall \sigma \in [\tau, s] \bigr\},
\end{align*}
where $\mathbb{U}$ is the set of all measurable control signals.~\hfill $\diamond$
\end{definition}
\begin{remark}
    In the literature, the sets in Definition~\ref{def:BRAS} are commonly referred to as backward reach-avoid \emph{tubes}, as the system may reach the target at any time $s \!\in \![\tau, 0]$ rather than at exactly the final time $0$. The term backward reach-avoid \emph{set} is reserved for the latter case. However, under the assumption of control invariant target sets, the two coincide~\cite{Begzadic25_SafeResets_Learning_HJB}, and we use the terms interchangeably throughout.
\end{remark}
The set $\mathrm{BRA}(\Xc^\star_j, \Oc, \tau)$ is analogous to the region of attraction~$\Rc_j$ considered in the stabilization case, while additionally accounting for a finite time horizon. We now extend the reach-avoid notion to the combinatorial setting.

\begin{definition}[($r$-out-of-$p$) Combinatorial BRA Sets]
\label{def:combo_BR}
Consider~\eqref{sys:ctrl_affine} with multiple target sets $\{\mathcal{X}_j^\star\}_{j\in\until p}$, an obstacle set $\mathcal{O}\subset\real^n$, and the control constraint set $\Uc\subseteq\real^m$. Given a negative time horizon $\tau<0$, let $ \mathrm{BRA}(\mathcal{X}_j^\star, \mathcal{O}, \tau)$ denote the BRA set associated with target $\mathcal{X}^\star_j$. With $r\in [p]$, the set:
\[
\mathcal{\tilde{R}}(\tau) := \Bigl\{ \bx \in \mathbb{R}^n \;\Big|\; \bigl|\{j \in [p] : \bx \in \mathrm{BRA}(\mathcal{X}^\star_j, \mathcal{O}, \tau)\}\bigr| \geq r \Bigr\},
\]
is called a \emph{($r$-out-of-$p$) combinatorial backward reach-avoid set} over the negative time horizon $\tau$.~\hfill $\diamond$
\end{definition}

\subsection{Combinatorial reach-avoid value functions}
To compute $\mathrm{BRA}(\mathcal{X}^{\star}_j, \mathcal{O}, \tau)$, we employ Hamilton-Jacobi reachability (HJR). We assume that the target sets $\{\mathcal{X}^\star_j\}_{j\in\until p}$ are control invariant and that the target sets and the obstacle set are defined with continuously differentiable functions $\ell_j : \mathbb{R}^n \to \mathbb{R}$ and $s : \mathbb{R}^n \to \mathbb{R}$ as:
\[
\mathcal{X}^{\star}_j = \{\bx \in \mathbb{R}^n \mid \ell_j(\bx) \ge 0\},
\quad
\mathcal{O} = \{\bx \in \mathbb{R}^n \mid \textcolor{black}{s(\bx) <} 0\}.
\]
For control-invariant target sets, the reach-avoid value function $V_j : \mathbb{R}^n \times (-\infty, 0] \to \mathbb{R}$ is the viscosity solution of the variational inequality \cite{Begzadic25_SafeResets_Learning_HJB}:
\begin{align}
&0 = \min\!\bigl\{s(\bx) - V_j(\bx,\tau),\;\tfrac{\partial V_j}{\partial \tau}(\bx,\tau) + H\bigl(\nabla_{\!\bx} V_j(\bx,\tau), \bx\bigr)\bigr\}, \notag\\
&V_j(\bx,0) = \min\{\ell_j(\bx),\, s(\bx)\},
\label{eq:hjb_pde}
\end{align}
for $\tau < 0$ and the control Hamiltonian 
$
    H(\boldsymbol{\lambda}, \bx) =
    \max_{\bu \in \mathcal{U}}\,
    \boldsymbol{\lambda}^\top
    \bigl(\bf(\bx) + \bg(\bx)\bu\bigr)
$~\cite{Fisac15_ReachAvoid, Begzadic25_SafeResets_Learning_HJB}. The value function encodes reach-avoid feasibility such that $V_j(\bx,\tau) \ge 0$ if and only if there exists a control that steers the system from state $\bx$ to $\mathcal{X}^{\star}_j$ while avoiding $\mathcal{O}$ over the horizon $[\tau, 0]$. That is:
\[
\textcolor{black}{\mathrm{BRA}(\mathcal{X}^{\star}_j, \mathcal{O}, \tau) = \{\bx \in \mathbb{R}^n \mid V_j(\bx, \tau) \ge 0\}.}
\]

Consider the goal of reaching a target set $\mathcal X_j^\star$ within a fixed horizon $T>0$ from an initial condition $\bx(0)\in \mathrm{BRA}(\mathcal X_j^\star,\mathcal O,-T)$. Then there exists a control that steers the system to $\mathcal X_j^\star$ by time $t=T$ while avoiding $\mathcal O$. As time progresses, the remaining horizon decreases, and so the set of states from which the target can still be reached safely within the remaining time shrinks accordingly. Therefore, along the trajectory, the state must remain in the time-varying set:
\[
\mathrm{BRA}(\mathcal X_j^\star,\mathcal O,\tau_1(t)),
\qquad \tau_1(t):=t-T,\quad t\in[0,T],
\]
Equivalently, one may introduce $\tau_1$ as an internal state satisfying $\dot\tau_1\!=\!1$ with the initial condition $\tau_1(0)\!=\!-T$, evolving linearly from $-T$ to $0$. As the BRA set is characterized by the value function, we realize our steering goal through rendering $V(\bx(t),\tau_1(t))$ nonnegative for time $t\!\in\![0,T]$.

In addition, we seek to ensure the state remains within the combinatorial BRA set throughout its trajectory. This guarantees that the system retains the ability to reach at least $r$ target sets at all times, thus preserving contingency options during the maneuver. 
Mathematically, the combinatorial BRA set can be expressed directly in terms of the value functions using order statistics $\tilde h(\bx,\tau) = \max^{(r)}\{V_j(\bx,\tau)\}_{j\in\until p}$:
\begin{equation}
\label{eq:hjb_ctilde}
\tilde \Rc(\tau) = \setdefb{\bx\in\real^n}{\tilde h(\bx,\tau)\geq 0}.
\end{equation}
Then, at all time $t\in[0,T]$, the system trajectory $t\!\mapsto\!\bx(t)$ is required to remain in $\tilde{\mathcal{R}}(\tau_2(t))$ where $t\!\mapsto \!\tau_2(t)$ captures how the time horizon for these contingency targets may change over time, potentially nonlinearly due to resource constraints such as available fuel or battery charge.  

We refer to this combined objective as \emph{combinatorial targeting}, which we pose as a controller design problem. Because the controller must track the evolving horizons~$\tau_1$ and $\tau_2$, which are not part of the physical state, combinatorial targeting naturally requires a dynamic feedback law. In particular, in our controller construction that follows, these horizons are incorporated directly as the internal states $\bz = (\tau_1, \tau_2)$. This motivates the following definition.
\begin{definition}
Consider a dynamic controller $(\bx,\bz)\!\mapsto\! \bk(\bx,\bz)$ equipped with an internal variable $\bz\!\in\!\real^q$ evolving under $\dot \bz\! =\! \bphi(\bx,\bz)$  with the internal dynamics $\map{\bphi\!}{\!\real^n\times\real^q\!}{\!\real^q}$. The controller $\bk$ \emph{($r$-out-of-$p$) combinatorially targets} a set $\mathcal{X}^\star_{j^\dagger}\! \in\! \{\mathcal{X}^{\star}_j\}_{j \in \until{p}}$ over a targeting horizon $T\!>\!0$ if any solution $t \!\mapsto\! \bx(t)$ of the closed-loop system under $\bu\!=\!\bk(\bx,\bz)$ satisfies:
\begin{itemize}
    \item if $\bx(0)\in\tilde\Rc(\tau_2(0))$, $\bx(t)\in\tilde \Rc(\tau_2(t))$ for all $t \in [0,T]$;
    \item if $\bx(0)\in\tilde{\Rc}(\tau_2(0)) \cap \mathrm{BRA}(\mathcal{X}^{\star}_{j^\dagger}, \mathcal{O}, -T)$, there exists $t^\star \in [0,T]$ such that $\bx(t^\star) \in \Xc^\star_{j^\dagger}$.
\end{itemize}
for a given time-varying negative horizon $t\mapsto \tau_2(t)$.~\hfill $\diamond$
\end{definition}

\subsection{Combinatorial targeting controllers}
In order to synthesize a filter analogous to~\eqref{eq:combo_clf_cbf_qp} we first make the following simplifying assumption:
\begin{assumption}
\label{ass:hjb_smooth}
For each $j \!\in\! [p]$ and a given time horizon $T>0$, $V_j\!: \mathbb{R}^n \!\times\! [-T, 0] \!\to\! \mathbb{R}$ is continuously differentiable on $\{(\bx, \tau) \!:\! V_j(\bx, \tau) \!\ge\! 0,\, \tau \!\in\! [-T, 0]\}$.
\end{assumption}
\begin{remark}[Technical gap in differentiability]
    Assumption~\ref{ass:hjb_smooth} is restrictive and, in general, not satisfied in practice. HJ value functions are typically only viscosity solutions of~\eqref{eq:hjb_pde} and are rarely continuously differentiable, requiring tools from nonsmooth analyses including generalized derivatives, and in our context, nonsmooth barrier function~\cite{Glotfelter17_NonsmoothCBF}. In particular, the viscosity CBF framework of~\cite{Hirsch2025ViscosityCBF} provides conditions under which forward invariance can be certified without requiring differentiability. However, incorporating such nonsmooth conditions into optimization-based controller synthesis remains challenging, as practical methods for computing and enforcing subdifferential constraints are still limited. Despite this limitation, HJR has been successfully applied in practice~\cite{Begzadic25_SafeResets_Learning_HJB} through numerical approximations. The purpose of this section is therefore to demonstrate the proposed combinatorial framework naturally extends to HJ-based constructions under a simplifying smoothness assumption, while a rigorous treatment in nonsmooth value functions remains an important direction for future work.~\hfill $\bullet$
\end{remark}

With differentiability of $V_j$, we can properly define a \textbf{combinatorial reach-avoid filter} via a QP:
\begin{align}
& {\small\begin{bmatrix}
    \bk(\bx,\tau_1, \tau_2)\!\\
    \boldsymbol{\theta}(\bx,\tau_1, \tau_2)\!
\end{bmatrix}} \!=\! \argmin_{\bu \in \mathcal{U},\,\boldsymbol{\omega}\in\realnonneg^2 }\;
\!\!\frac{1}{2}\|\bu \!-\! \bu_{\mathrm{nom}}(\bx,\tau_1)\|^2\! +\!
d_\omega\Vert\boldsymbol{\omega}^2\Vert \nonumber\\
& \text{s.t. }\;\forall\, j \in [p],\;\label{eq:combo_qp_hjr}\\
&  \dot{V}_{j^\dagger}(\bx,\bu,\tau_1)\!\ge\!
-\!\alpha_{j^{\dagger}}\!\bigl(V_{j^\dagger}(\bx,\tau_1)\!\bigr)\! \!-\! \omega_1\mathrm{ReLU}\bigl(\!-\alpha_{j^{\dagger}}(V_{j^\dagger}(\bx,\tau_1)\!)\!\bigr)
\nonumber\\
&  \dot{V}_j(\bx,\bu,\tau_2)\!\ge\!
-\alpha_{\tilde{h}}\bigl(V_j(\bx,\tau_2)\bigr)\!
\!-\! \omega_2\rho\bigl(V_j(\bx,\tau_2)\! \!-\! \tilde{h}(\bx,\tau_2)\bigr), 
\nonumber
\end{align}
where $\dot{V}_j(\bx,\bu,\tau)\! =\! \dot{\tau}\,\partial_{\tau} V_{j}(\bx,\tau)
\!+\! L_{\bf}V_{j}(\bx,\tau)
\!+ \!L_{\bg}V_{j}(\bx,\tau)\bu$, $d_\omega > 0$, $\alpha_{j^{\dagger}}, \alpha_{\tilde{h}} \in \mathcal{K}_e$, and $\rho : \mathbb{R} \to \mathbb{R}_{\ge 0}$ is positive definite. The internal dynamics are given by $\dot\tau_1=1$ and $\dot\tau_2 = \frac{d\tau_2}{dt}(t)$ for a given $t\mapsto \tau_2(t)$ (with a slight abuse on the notation).

Based on the above formulation, the following  result establishes combinatorial targeting.
\begin{proposition}\label{thm:combo_HJR}
    Consider the control-affine system~\eqref{sys:ctrl_affine} with a collection of reach-avoid value functions $\{V_j\}_{j\in\until p}$ associated with control invariant target sets $\{\mathcal{X}^{\star}_j\}_{j\in\until p}$, obstacle set $\mathcal{O}$, and compact control constraint set $\mathcal{U}$. Fix an index $j^\dagger\in\until p$. In addition to Assumption~\ref{ass:hjb_smooth}, suppose that the QP controller defined in~\eqref{eq:combo_qp_hjr} is strictly feasible for all $\bx\!\in\!\tilde \Rc(\tau_2(t))$ and all $t\in[0,T]$, under a given time-varying $\map{\tau_2\!}{\![0,T]\!}{\real_{<0}}$. Then, the controller $\bk$ ($r$-out-of-$p$) combinatorially targets $\mathcal{X}^{\star}_{j^\dagger}$ over a targeting horizon $T>0$.
\end{proposition}
\begin{proof}
    In order to apply Lemma~\ref{lem:combo_CBF} to the time-varying set~$\tilde R(\tau_2(t))$, we consider an augmented closed-loop system:
    $$
    \dot{\bx} = \bf(\bx) + \bg(\bx)\bk(\bx,\tau_1,\tau_2),~\dot \tau_1 = 1,~\dot \tau_2 =  \frac{d\tau_2}{dt}(t),~\dot t = 1
    $$
    with state $(\bx,\tau_1,\tau_2,t)$. 
    Under this representation, $\tilde \Rc(\tau_2(t))$ becomes a time-invariant set in augmented state space, and Lemma~\ref{lem:combo_CBF} guarantees 
    $\bx(t)\!\in\!\tilde \Rc(\tau_2(t))$ for all $t\! \in\! [0,T]$ whenever $\bx(0)\!\in\!\tilde\Rc(\tau_2(0))$.\footnote{A complete justification extending Lemma~\ref{lem:combo_CBF} to time-varying sets requires characterization of forward invariance over time interval. We omit these technical details for brevity.}
    We now show that $\bk$ steers the state into $\mathcal{X}^\star_{j^\dagger}$ in finite time if $\bx(0)\!\in\!\tilde{\Rc}(\tau_2(0)) \!\cap\! \mathrm{BRA}(\mathcal{X}^{\star}_{j^\dagger},\Oc,T)$.

    For all $\tau_1< 0$ and all $\bx\!\in\!\tilde\Rc\cap \mathrm{BRA}(\mathcal{X}^{\star}_{j^\dagger},\Oc,\tau_1)$, the $\mathrm{ReLU}$ term in~\eqref{eq:combo_qp_hjr} vanishes since $V_{j^\dagger}(\bx,\tau_1)\!\geq\! 0$. Hence, a standard CBF inequality as in~\eqref{eq:cbf_cond} holds for the function $V_{j^\dagger}$. In addition, because the $\mathrm{ReLU}$ term is zero, the optimal choice of $\omega_1$ in the QP~\eqref{eq:combo_qp_hjr} is zero. Hence, from the continuity of the QP controller (under strict feasibility, i.e., Slater's condition  assumption~\cite{PM-AA-JC:25}), there is a neighborhood $\mathcal N$ of $\Rc\cap \mathrm{BRA}(\mathcal{X}^{\star}_{j^\dagger},\Oc,\tau_1)$ where $\theta_1(\bx,\tau_1,\tau_2)< \epsilon$ for all $\bx\in\mathcal{N}$. Therefore, the standard barrier condition holds for $V_{j^\dagger}$  on some neighborhood $\cal N$, with a $\Kc_e$ function defined as $(1-\epsilon)\alpha(s)$ for $s<0$ and $\alpha(s)$ for $s\ge 0$ . Hence, for any system trajectory $t\mapsto \bx(t)$ initialized with $\bx(0)\in\tilde\Rc\cap \mathrm{BRA}(\mathcal{X}^{\star}_{j^\dagger},\Oc,\tau_1(0))$, we have $V_{j^\dagger}(\bx(t), t\!-\!T) \geq 0$ for all $t\! \in \![0, T]$, where we have substituted $\tau_1$ using the definition $\tau_1(t)=t-T$.
    In other words, $\bx(t)\!\in\! \mathrm{BRA}(\Xc^\star_{j^\dagger},\Oc, t\!-\!T)$ at all time $t\in[0,T]$, and therefore, $\bx(t)\notin\Oc$. In addition, at time $t=T$ in particular, the terminal condition~\eqref{eq:hjb_pde}, $V_{j^\dagger}(\bx(T), 0) \!=\! \min\{\ell_{j^\dagger}(\bx(T)), s(\bx(T))\}\! \geq\! 0$ implies $\ell_{j^\dagger}(\bx)\geq 0$, and hence, $\bx(T) \in \mathcal{X}^\star_{j^\dagger}$. 
    This establishes that $\bk$ ($r$-out-of-$p$) combinatorially targets $\mathcal{X}^\star_{j^\dagger}$.
\end{proof}
Proposition~\ref{thm:combo_HJR} shows that, despite the combinatorial nature of the problem, combinatorial targeting can be achieved with a combinatorial reach-avoid filter that involves only $p+1$ constraints. More importantly, the proposed formulation naturally enables real-time switching between target sets. At any time, the active index $j^\dagger$ may be updated to one of the indices certified by the combinatorial BRA~$\tilde \Rc(\tau_2(t))$, allowing the controller to adapt online. 
Note that at such switching time $t_{\rm s}$, the internal variable $\tau_1$ is updated to $\tau_2$ to ensure feasibility, and the contingency target can be reached within the time period $[t_{\rm s},t_{\rm s}+\tau_2(t_{\rm s})]$.

\begin{figure}
    \centering
    \includegraphics[width=1\linewidth]{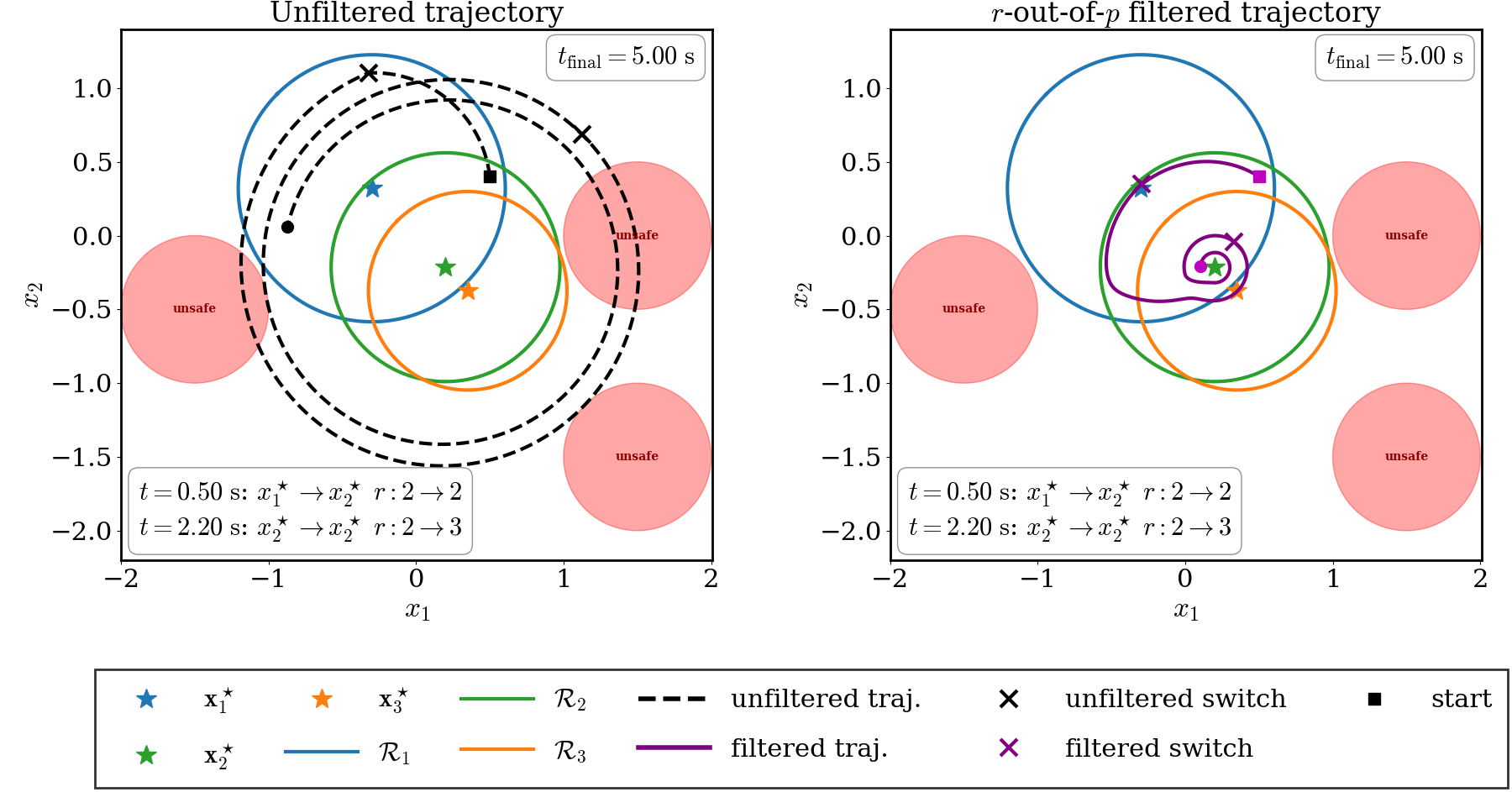}
    \caption{Example 1: Linear system with $p=3$ targets. The simulation starts with $j^\dagger=1$ and $r=2$, switches to $j^\dagger=2$ at $t=0.5\,\mathrm{s}$, and later raises $r$ to $3$. The filtered trajectory remains within $r$-out-of-$p$ safe sets throughout, while the nominal trajectory violates obstacle constraints.\vspace*{-10pt}}
    \label{fig:linsystem}
\end{figure}

\section{Simulations \label{section:results}}
\subsection{Example 1: Linear System (CLF-based)}
 Consider the linear system with $\bx\!\in\!\mathbb{R}^2$, $\bu\!\in\!\mathbb{R}$:
\begin{align*}
\dot \bx \!=\! A\bx+B\bu,\;\;  
A=\scalebox{0.85}{$\begin{bmatrix}0.9&-3.0\\[2pt]4.0&-0.1\end{bmatrix}$},
\;\,
B=\scalebox{0.85}{$\begin{bmatrix}1\\[1pt]1\end{bmatrix}$},
\end{align*}
Three target equilibria $\bx_j^\star$, $j\!\in\![3]$, and three obstacle sets
$\mathcal O_\ell\!=\!\{\bx\in\mathbb R^2:\|\bx\!-\!\bp_q\|_2\le R_q\}$, $q\!=\!1,2,3$ are given. The nominal controller for each  selected target $\bx_{j^\dagger}^\star$ is: 
\begin{equation*}
\bu_{\mathrm{nom}}(\bx;\bx_{j^\dagger}^\star)=\bu_{j^\dagger}^\star - K(\bx-\bx_{j^\dagger}^\star),\quad
 K=\begin{bmatrix}1&\!0\end{bmatrix},
\end{equation*}
where $\bu_{j^\dagger}^\star$ satisfies $A\bx_{j^\dagger}^\star \!+\! B\bu_{j^\dagger}^\star \!=\! \mathbf{0}$. We define CLFs $V_j(\bx)\!:=\!(\bx\!-\!\bx_j^\star)^\top P(\bx\!-\!\bx_j^\star)$ with
positive definite $P\in\real^{n\times n}$ from $(A\!-\!BK)^\top P \!+\! P (A\!-\!BK)\!=\!-I$ and $c_j\!:=\!\nu \min_{q}\;\min_{\|\bx\!-\!\bp_q\|_2=R_q} V_j(\bx)$ 
for $\nu\in(0,1)$, so that $\Rc_j$ is the largest obstacle-free CLF sublevel set centered at $\bx_j^\star$. The proposed combinatorial stabilization filter~\eqref{eq:combo_qp_hjr} is implemented\footnote{with $\nu\!=\!0.9$, $d_\omega\!=\!0.1$, $\alpha_{j^\dagger}(s)\!=\!2s$, $\alpha_{\tilde{h}}(s)\!=\!0.18s$, and $\rho(s)\!=\!\alpha_{\tilde{h}}\,s^2$, $\bx_1^\star\!=\![-0.30,\, 0.321]^\top$, $\bx_2^\star\!=\![0.20,\, -0.214]^\top$, $\bx_3^\star\!=\![0.35,\, -0.374]^\top$, $\bp_1\!=\![-1.5,\, -0.5]^\top$, $\bp_2\!=\![1.5,\, 0]^\top$, $\bp_3\!=\![1.5,\, -1.5]^\top$, $R_\ell\!=\!0.5$.} using  CVXPY~\cite{agrawal2018rewriting} and the CLARABEL solver  ~\cite{Clarabel_2024}.
The resulting simulation in Fig.~\ref{fig:linsystem} is initialized with $\bu_{\textrm{nom}}$ and the CLF constraint targeting $\bx_1^\star$ and with the combinatorial requirement $r\!=\!2$. Consequently, the filter preserves membership in the regions of attraction of at least $r$ candidate targets while steering the state toward the selected one, whereas the unfiltered trajectory repeatedly violates safety despite using a stabilizing controller. This behavior is already evident in the start of the simulation as the nominal trajectory is immediatelly modified to stay within reach of a second target, effectively hugging the boundary of the second active safe set. 
At time $t\!=\!0.5\,\mathrm{s}$, the desired equilibrium target is switched to $\bx_2^\star$. As the filtered trajectory was kept within the reach-avoid sets of two targets, this switch can be executed safely, without collision, whereas the unfiltered trajectory intersects the obstacle. A subsequent increase of $r$ to $3$, illustrating the flexibility of the combinatorial formulation.

\begin{figure*}[!t]
    \centering
    \includegraphics[width=0.98\linewidth]{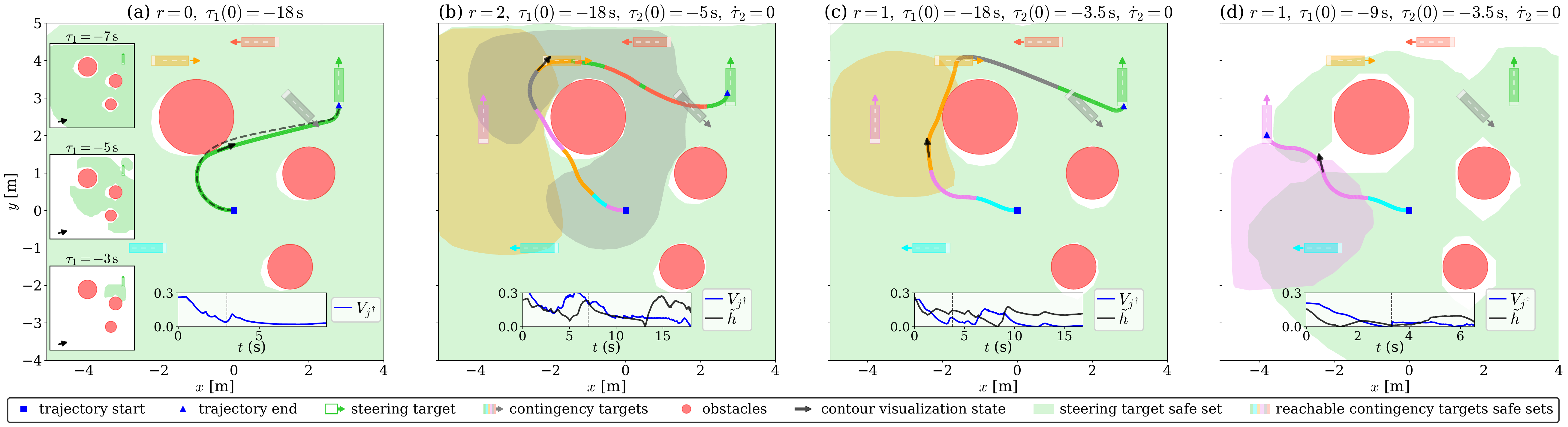}
    \vspace{-4pt}
    \caption{\textcolor{black}{
    Example 2: Simplified aircraft with $p=6$ runway targets navigating between obstacles. Dashed line shows the  unfiltered trajectory, while solid lines show the filtered trajectories colored by a currently active contingency target. 
    Four cases are presented with varying $r$ and horizon times $\tau_1(0), \tau_{2}(0)$. 
    Reach-avoid $0$-superlevel sets are plotted at the state (and time) indicated by the arrow (and dotted line in inset).
    Insets in (a) show the evolution of the steering reach–avoid 0-superlevel set for the green runway for decreasing horizon $\tau_1$.}
    \vspace*{-11pt}}
    \label{fig:nonlinear_HJR}
\end{figure*}

\subsection{Example 2: Airplane Landing (HJR-based)}
\noindent Consider a simplified aircraft model with elevation control:
\begin{align*}
\dot x = v\cos\theta,\;\; \dot y = v\sin\theta,\;\; \dot z = u_3,\;\; \dot\theta = u_1,\;\; \dot v = u_2 - c_v v,
\end{align*}
with control input $\bu  \!\in \!\mathcal{U}\! =\! [-1,1]^3$, where $u_1$ is the angular rate, $u_2$ the forward thrust, $u_3$ the vertical velocity, and $c_v \!=\! 0.3$ is a drag coefficient.
The task is to navigate between building obstacles towards one of $p\!=\!6$ runway targets. Each runway target is defined by a position $(x,y)$ at elevation $z\!=\!0$ and an approach angle~$\theta$. The nominal controller minimizes the planar distance to the steering target until nearby, at which point it minimizes the angle offset and decreases the velocity $v$ and elevation $z$ to zero. We consider a scenario where contingencies have to be reachable within a fixed horizon, i.e. $\tau_{2}(0)\!=\!c\!<\!0$ and $\dot{\tau}_{2}\!=\!0$, while the selected target set is shrinking throughout the trajectory.
The proposed combinatorial reach-avoid filter~\eqref{eq:combo_qp_hjr} is applied with reach-avoid value functions computed via HJR.\footnote{for $\bx_1^\star\!=\![3,3, \pi / 2], \> \bx_2^\star\!=\![-2,-1, \pi], \> \bx_3^\star\!=\![-2, 4,0], \> \bx_4^\star=[-4, 2, \pi/2], \> \bx_5^\star\!=\![1, 3, -\pi/4], \> \bx_6^\star\!=\![4.5,2,\pi]$. $\alpha_{j^\dagger}(s)\!=\!2s, \alpha_{\tilde{h}}(s)=s$, $\mathcal{O}_1=B((2,1), 0.7), \> \mathcal{O}_2=B((-1, 2.5), 1), \> \mathcal{O}_3=B((1.5, -1.5), 0.6)$.}
Figure~\ref{fig:nonlinear_HJR} illustrates four scenarios of steering to the green runway, with subplots showing how the reach-avoid target set shrinks with the decreasing available time horizon $\tau_1$. Without contingency requirements in a), the nominal trajectory collides with an obstacle, while the filtered trajectory reaches the target safely. Adding $r\!=\!2$ contingencies with $\tau_{2}(0)\!=\!-5$\,s in b) forces the filtered trajectory to keep two runways reachable throughout. Reducing to $r=1$ with a shorter horizon of $\tau_{2}(0)\!=\!-3.5$\,s in c) produces a more direct trajectory. When the available time to land $\tau_1(0)$ is shortened in d), the original target becomes infeasible and the filter switches $j^\dagger$ to a contingency runway (pink). In all four scenarios, the inset plots confirm that the steering and pivot value functions remain nonnegative along the filtered trajectory, certifying that the combinatorial reach-avoid guarantee is maintained.


\section{CONCLUSIONS }
In this paper we formalized the notion of pursuing a target while retaining contingency options, and provided tractable QP-based filters with $p\!+\!1$ constraints for enforcing it. For stabilization problems, we encoded CLF-based steering and combinatorial barrier conditions to guarantee combinatorial stabilization. For finite-horizon or resource depletion problems, we developed a parallel framework using HJR value functions that guarantees targeting while reataining $r$-out-of-$p$ contingency targets. Both frameworks were validated on examples demonstrating safe target switching where nominal controllers alone violate safety. Future work will include disturbances and investigate the nonsmoothness in the HJR value functions, further expanding the applicability of the proposed combinatorial framework.

\bibliographystyle{IEEEtran}
\bibliography{ifacconf}

\end{document}